\documentclass[11pt]{article}
\usepackage[dvipsnames]{xcolor}
\usepackage{amssymb,amsfonts,amsmath,amsthm,amscd,dsfont}

\usepackage{fullpage}
\usepackage[colorlinks]{hyperref}
\usepackage[capitalize, nameinlink]{cleveref}

\usepackage{tikz}

\hypersetup{
  linkcolor=[rgb]{0,0,0.4},
  citecolor=[rgb]{0, 0.4, 0},
  urlcolor=[rgb]{0.6, 0, 0}
}

\usepackage{complexity}

\newcommand{\F}{\mathbb{F}}
\newcommand{\N}{\mathbb{N}}
\newcommand{\Z}{\mathbb{Z}}

\newcommand{\Ext}{\mathsf{Ext}}
\newcommand{\Int}{\mathsf{Int}}

\newcommand{\set}[1]{\left\{ #1 \right\}}

\newcommand{\bool}[1]{{#1}_{\text{bool}}}

\newtheorem{theorem}{Theorem}[section]
\newtheorem{lemma}[theorem]{Lemma}
\newtheorem{definition}[theorem]{Definition}

\newtheorem{claim}[theorem]{Claim}

\newtheorem{corollary}[theorem]{Corollary}

\title{A Lower Bound for Read-Once Parity Branching Programs}
\author{Ben Lee Volk\thanks{Efi Arazi School of Computer Science, Reichman University, Israel. Email: \texttt{benleevolk@gmail.com}. The research leading to these results has received funding from the  Israel Science Foundation (grant number 843/23). }}

\date{}

\begin{document}

\maketitle

\abstract{
We prove an $\tilde{\Omega}(n^2)$ lower bound for read-once parity branching programs computing an explicit boolean function on $n$ variables. The previous best lower bound was $\tilde{\Omega}(n^{1.5})$. Our lower bound is proved by reducing the problem to a lower bound in algebraic circuit complexity.}

\section{Introduction}

Algebraic complexity is a beautiful and mathematically rich area that studies the complexity of symbolic computation of polynomials. Virtually all the known algorithms for algebraic problems (such as computing the determinant or permanent, multiplying matrices, or computing the discrete Fourier transform) are naturally modeled using algebraic models. One of its raisons d'être, however, is also the hope that lower bounds in the algebraic model will inspire lower bounds in the arguably more natural, and definitely more common, \emph{boolean} models of computation. A long line of work on lower bounds for algebraic models has had numerous successes, such as, to give a non-exhaustive list, super-polynomial lower bounds for monotone circuits \cite{JS82}, non-commutative formulas \cite{Nisan91}, multilinear formulas \cite{Raz09, RY09}, and bounded-depth circuits \cite{LST25, Forbes24}; and super-linear lower bounds for circuits \cite{Str73,BS83}, algebraic branching programs and formulas \cite{CKSV22, Kalorkoti85}. More comprehensive surveys of lower bounds in algebraic complexity are \cite{Sap15, SY10}.

These lower bounds use the \emph{syntactic} nature of the computation. For some, it is not clear what the analogous boolean model is, and for some the corresponding lower bounds for boolean models have been in fact known even earlier.

Motivated by considerations from proof complexity, there has been some work on \emph{functional} lower bounds for algebraic circuits \cite{GR00, FSTW21, FKS16, HLT24}. These are lower bounds for algebraic models that do not apply only to a single polynomial, but rather to a set of polynomials all computing the same function over some limited domain.

In this paper, we give an instance in which one can prove a lower bound on a bona fide boolean model of computation by reducing to a lower bound on an algebraic model of computation. One of the main obstacles to obtaining lower bounds on boolean circuits using lower bounds on algebraic circuits is that boolean circuits can exploit boolean identities that do not hold in the algebraic setting.
One can trivially convert a boolean circuit $C$ computing a function $g : \set{0,1}^n \to \set{0,1}$ to an algebraic circuit over $\F_2$ gate-by-gate (say by replacing AND gates with multiplication gates and NOT gates with gates that add, modulo 2, the boolean value `$1$'), and the resulting algebraic circuit computes a polynomial that agrees with $g$ on the boolean cube. But its specific form depends on the circuit $C$: as a trivial example, the boolean function $g(x)=x$ is functionally identical to the function $g(x)=x \wedge x$, but the straightforward way alluded to above for converting a boolean circuit computing $x \wedge x$ to a polynomial would result in the polynomial $x^2$, which is distinct from the polynomial $x$. Therefore, a lower bound on algebraic circuits computing a specific polynomial doesn't rule out the possibility that there's a different efficient way to compute the same function over the boolean domain.

The driving force behind our method is that some boolean models of computation yield \emph{multilinear} polynomials when one applies the natural transformation that ``algebrizes'' them. Since two multilinear polynomials that agree on $\F_2^n$ are identical, we can deduce exactly which polynomial is obtained after this transformation, and prove (syntactic) lower bounds for this polynomial. The easy proofs for these observations appear in \cref{sec:alg-to-bool}.

\subsection{Read-Once Parity Branching Programs}

The model we consider is the following:

\begin{definition}
\label[definition]{def:parity-BP}
A read-once parity branching program ($\oplus$-BP) is a directed, acyclic multigraph with a source node $s$ and a target node $t$. Every edge in the graph is labeled either by a constant in $\set{0,1}$, a variable $x_i$ or a negated variable $\neg x_i$. On every $s \to t$ path, every variable appears at most once. The program accepts an input $x$ if the number of $s \to t$ paths consistent with $x$ is odd. The \emph{size} of the program is the number of edges.
\end{definition}

Parity branching programs have been considered as a natural extension of deterministic branching programs, and a natural variant of non-deterministic branching programs (see Part V of the book \cite{Juk12} for a thorough survey on this area).
In the deterministic and non-deterministic models, exponential lower bounds for read-once branching programs have been known for decades \cite{Zak84, BHST87, Wegener88} and there are even lower bounds for branching programs that are allowed to read every input at most $k$ times along every path for any constant $k$ \cite{Oko91, BRS93, Thathachar98}.
Note that in the definition above the restriction is syntactic: we require that on every path, every variable appears at most once. The ``semantic'' model that only imposes this condition on paths consistent with some input has also been considered in the branching program literature (in \cite{Juk12} it is called \emph{weakly} read-once, and it is shown to be exponentially more powerful than the syntactic model).

However, the lower bounds for the deterministic and non-deterministic models do not apply in the parity branching program model. Jukna \cite[Research Problem 16.14]{Juk12} explicitly poses the question of proving exponential lower bounds for read-once $\oplus$-BPs.
Jukna proves an exponential lower bound when the branching program is \emph{oblivious}. An oblivious read-once $\oplus$-BP is a read-once branching program which is also layered, and in every layer all edges are labeled using the same variable. Prior to this work, the best lower bound for read-once $\parity$-BPs was $\Omega(n^{3/2} / \log n)$, for the element distinctness function, which follows by adapting Nečiporuk's \cite{Nec66} method to this model (the argument appears in \cite{KW93}, where it is attributed to Pudlák). This lower bound in fact holds for general parity branching programs, even without the read-once restriction. Cheraghchi, Hirahara, Myrisiotis and Yoshida proved a similar lower bound for the meta-complexity problem $\mathsf{MKTP}$ \cite{CHMY24}, and this model was also studied by Homeister \cite{SavickyS05, Homeister06, BHW03}, who proved lower bounds in some restricted settings.

Our main result is an almost quadratic lower bound for read-once parity branching programs.

\begin{theorem}
\label[theorem]{thm:main}
There exists an explicit family of functions $\set{f_n : \set{0,1}^n \to \set{0,1}}_{n \in \N}$ such that any read-once parity branching program computing $f_n$ has size $\Omega(n^2/\log^2 n)$. 
\end{theorem}

Here ``explicit'' means that  there exists a polynomial-time Turing machine $M$ that on input $x=(x_0, \ldots,x_{n-1})$ computes $f_n (x)$.

\medskip

Note that in \cref{def:parity-BP} we didn't insist that the graph is layered. Jukna \cite{Juk12} similarly does not require branching programs to be layered graphs. Indeed, when branching programs are used as a model of computation in a circuit complexity context (as opposed to using them to model space-bounded computation of Turing machines), requiring them to be layered imposes a rather artificial constraint. Of course, any branching program can be made to be layered with a polynomial blow-up, so this wouldn't have mattered if we could prove super-polynomial lower bounds. In fact, if we assume that the graph is layered, we can prove a slightly better $\Omega(n^2)$ lower bound much more easily. We provide the details in \cref{sec:layered}.

We further remark that (again, as in \cite{Juk12}), we measure the size of the program by the number of edges: since the in-degree of any vertex is unbounded, this is a natural complexity measure. This is another distinction that is only important insofar as the lower bounds we can prove are merely polynomial.

\subsection{Technique}

As mentioned above, we observe a natural connection between the problem of proving lower bounds for read-once $\oplus$-BPs and 
 the problem of proving lower bounds for multilinear algebraic branching programs, a long-standing problem in algebraic complexity theory (see, for example, the recent works \cite{CKSS24, FLSY26}).
 
A natural ``algebrization'' operation on parity branching programs computing a boolean function $h$ results in a syntactic multilinear branching program, a well-studied model in algebraic complexity (see \cref{sec:alg-to-bool}). Further, applying this operation on any read-once parity branching program would give an algebraic branching program computing the unique multilinear polynomial that agrees with $h$ on $\F_2^n$.
 
 To prove our lower bound, we resort to a result of Alon, Kumar and the author \cite{AKV20} which proves such a lower bound for the model of syntactically multilinear \emph{circuits} (improving an earlier result of \cite{RSY08}). Circuits are stronger than branching programs, but when one deals with lower bounds that are merely super-linear (rather than super-polynomial) one has to be a bit careful when defining the model. However it turns out that the lower bound does apply to the algebraic branching programs that are obtained by converting $\oplus$-BPs to algebraic models. These models are defined in \cref{sec:alg-to-bool}.
 
 The last remaining ingredient then is to prove that the family of functions for which the lower bound of \cite{AKV20} applies is explicit (in the sense of being in $\P$). This does not follow immediately from the results of \cite{AKV20} (as the definition of their polynomial involves a sum over a set of exponential size), and requires some work. We give a dynamic programming algorithm that computes this function. This algorithm appears in \cref{sec:dp}.
 
 As a by-product, we also slightly tighten the results of \cite{AKV20} and prove a lower bound for a polynomial in $\VP$ (the lower bound in \cite{AKV20} was claimed for a polynomial in $\VNP$). The details appear in \cref{sec:VP}.

\section{Syntactically Multilinear ABPs and Read Once Parity Branching Programs}
\label[section]{sec:alg-to-bool}

We start by defining the algebraic analog of read-once $\oplus$-BPs.

\begin{definition}
\label[definition]{def:ABP}
A syntactically multilinear algebraic branching program (ABP) over $\F_2$ is a directed, acyclic multigraph with a source node $s$ and a target node $t$. Every edge in the graph is labeled either by a constant in $\F_2$ or a linear function in some $x_i$. On every $s \to t$ path, every variable appears at most once. Each $s \to t$ path computes the product of the labels on the path, and the program computes the sum, over all $s \to t$ paths, of the polynomials computed by the paths. The \emph{size} of the program is the number of edges.
\end{definition}

Note that unlike some common definitions in the literature, we didn't allow edges to be labeled by arbitrary linear functions in the variables, but rather only by a constant or a linear function in a single variable.

Similarly, an algebraic circuit $C$ is called \emph{syntactically multilinear} if every multiplication gate in $C$ multiplies two variable-disjoint subcircuits. Circuits are stronger than ABPs:

\begin{claim}
\label[claim]{cl:abp-to-ckt}
Suppose $f \in \F_2[x_1, \ldots, x_n]$ is computed by a syntactically multilinear ABP of size $S$. Then $f$ is computed by a syntactically multilinear circuit of size $O(S)$.
\end{claim}

\begin{proof}
Simulate the ABP vertex by vertex. For a vertex $v$, the polynomial computed by $v$, denoted $f_v$, is defined to be the polynomial computed by the sub-ABP whose source is $s$ and sink is $v$.
By induction from $s$, for every vertex $v$ in the ABP we add a sum gate $v'$ to the circuit computing $f_v$. If $v$ has incoming edges from vertices $u_1, \ldots, u_k$ with edge labels $\ell_1, \ldots, \ell_k$, the sum gate $v'$ computes $\sum_{i=1}^k \ell_i f_{u_i}$, where the gates $u'_i$ computing $f_{u_i}$ have been already added to the circuit by induction, and $\ell_i$ is a linear function in a single variable and can be computed by a circuit of size $O(1)$.

For every edge in the ABP we need to add $O(1)$ edges to the circuit. Since the ABP is syntactically multilinear, every product gate multiplies a linear function in a variable $x_i$ by a subcircuit $C$ in which $x_i$ doesn't appear, so the circuit is syntactically multilinear.
\end{proof}

We now construct a polynomial that requires syntactically multilinear circuits (and hence syntactically multilinear ABPs) of size $\Omega(n^2/\log^2 n)$.

In what follows, we associate $\Z_n = \Z / (n)$ with the set $\set{0,1,\ldots,n-1}$ with addition modulo $n$. However, we sometimes think of the elements of $\Z_n$ as integers under the natural ordering.

\begin{definition}
\label[definition]{def:hard-poly}
Let $\F$ be a field and $n$ an even integer.
Let $B \subseteq \Z_n$ be a subset of size $n/2$. Construct a bijection $\sigma_B : B \to \Z_n \setminus B$ in the following manner: think of the $n$ elements of $\Z_n$ arranged on a cycle in a clockwise direction. Starting from $0$ and going clockwise, pick the first element $j \in B$ such that the next element on the cycle, $k$, is not in $B$. Define $\sigma_B(j)=k$, erase $j$ and $k$ from the cycle and continue in that manner until all elements in $B$ are assigned values. Let $f_B (x_0, \ldots, x_{n-1}) = \prod_{j \in B} (x_j + x_{\sigma_B(j)})$. Finally, define the following polynomial in $\F[x_0, \ldots, x_{n-1}, y_0, \ldots, y_{n-1}]$:
\[
f(x_0, \ldots, x_{n-1}, y_0, \ldots, y_{n-1}) = \sum_{\substack{B \subseteq \Z_n \\ |B|=n/2}} \prod_{j \in B} y_j \cdot f_B(x_0, \ldots, x_{n-1}).
\]
\end{definition}

This polynomial was constructed by Raz, Shpilka and Yehudayoff \cite{RSY08}, who proved a super-linear lower bound on the size of syntactically multilinear circuits computing it. This lower bound was improved in \cite{AKV20}:

\begin{theorem}[\cite{RSY08, AKV20}]
\label[theorem]{thm:lb-ckt}
Let \[
f(x_0, \ldots, x_{n-1}, y_0, \ldots, y_{n-1}) \in \F[x_0, \ldots, x_{n-1}, y_0, \ldots, y_{n-1}]
\]
be defined as above. Any syntactically multilinear circuit computing $f$ has size $\Omega(n^2 / \log^2 n)$.
\end{theorem}

\cref{thm:lb-ckt} holds over any field, in particular over $\F_2$. We remark that the proofs in \cite{RSY08, AKV20} don't use the special structure of the bijection $\sigma_B$ outlined above, but rather only need $\sigma_B$ to be some bijection from $B$ to $\Z_n \setminus B$. This structure however will come in handy in \cref{sec:dp}.

The following corollary follows immediately from \cref{cl:abp-to-ckt}.

\begin{corollary}
\label[corollary]{thm:lb-ABP}
Any syntactically multilinear ABP computing $f$ has size $\Omega(n^2 / \log^2 n)$.
\end{corollary}

\subsection{Connections between the Algebraic Model and the Boolean Model}

For a polynomial $f \in \F_2[x_0, \ldots, x_{n-1}]$, let $\bool{f} : \set{0,1}^n \to \set{0,1}$ be the boolean function that $f$ represents.
It is clear that an \emph{upper bound} on the ABP complexity of $f$ gives an upper bound on the $\parity$-BP complexity of $\bool{f}$, by ``booleanizing'' the ABP: replacing every label $1+x_i$ by $\neg x_i$ and treating the new graph as a boolean $\oplus$-BP gives a $\oplus$-BP that computes the same function as $f$ on any inputs in $\set{0,1}^n$, and therefore correctly computes $\bool{f}$.

This observation shows that solving Jukna's  \cite[Research Problem 16.14]{Juk12} would have major consequences in algebraic complexity: indeed, proving such a lower bound for $\bool{f}$ would show a lower bound on the syntactically multilinear ABP size of $f$. Furthermore, since syntactically multilinear circuits can be simulated by formulas (and hence ABPs) with a quasi-polynomial blow-up \cite{RY08}, such a result would even imply an exponential lower bound on syntactically multilinear \emph{circuits}.

We remark that Jukna's oblivious model corresponds to a model called read-once oblivious ABPs which was well-studied in algebraic complexity (see, e.g., \cite{FS13}).

In the read-once setting, one could also deduce boolean complexity \emph{lower bounds} from algebraic complexity lower bounds. The following claim is incredibly simple, but it is perhaps the key point behind our lower bound.

\begin{claim}
\label[claim]{cl:alg-to-bool}
Let $f \in \F_2[x_0, \ldots, x_{n-1}]$ be a multilinear polynomial. If $f$ requires syntactically multilinear ABPs of size $S$, then $\bool{f}$ requires read-once $\oplus$-BPs of size $S$.
\end{claim}

\begin{proof}
Consider any read-once $\oplus$-BP computing $\bool{f}$ of size $S'$. ``Algebrize'' the branching program by replacing every label $\neg x_i$ by $1+x_i$, and treating it as a syntactically multilinear ABP. This ABP computes a multilinear polynomial $g \in \F_2[x_0, \ldots, x_{n-1}]$ that agrees with $f$ on $\F_2^n$, in the sense that for every $x$, $g(x)=f(x)$. Since $f$ and $g$ are both multilinear, $g=f$. Thus, we obtained a syntactically multilinear ABP computing $f$, which implies by our assumption that $S'\ge S$.
\end{proof}

\begin{corollary}
\label[corollary]{cor:bool-lower-bound}
Let $f$ be as in \cref{def:hard-poly}. Then $\bool{f}$ requires read-once $\oplus$-BPs of size $\Omega(n^2 / \log^2 n)$.
\end{corollary}

\begin{proof}
Follows from \cref{cl:alg-to-bool} and \cref{thm:lb-ABP}.
\end{proof}

Using nothing more than the definition in \cref{def:hard-poly}, one could show that the function $\bool{f}$ is in the class $\parity\P$.

\begin{claim}
$\bool{f} \in \parity\P$.
\end{claim}

\begin{proof}
Consider a non-deterministic TM $M$ that, on input $x_0, \ldots, x_{n-1}, y_0, \ldots, y_{n-1}$, guesses a subset $B \subseteq \Z_n$. If $|B| \neq n/2$, $M$ rejects. Otherwise, $M$ treats its input as elements in $\F_2^{2n}$, computes (using the notations of  \cref{def:hard-poly})
\[
\prod_{j \in B} y_j \cdot f_B
\]
and accepts iff the result is 1. This last computation can be done in deterministic polynomial time. Then, on input $x_0, \ldots, x_{n-1}, y_0, \ldots, y_{n-1},$ $M$ has an odd number of accepting paths iff $\bool{f}(x_0, \ldots, x_{n-1}, y_0, \ldots, y_{n-1})=1$.
\end{proof}

However, this is not entirely satisfying. Usually, in the context of circuit lower bounds, one would like to prove lower bounds for \emph{explicit} functions, namely, functions in $\P$ (or $\NP$). In \cref{sec:dp} we prove that the function $\bool{f}$ is actually explicit in that exact sense.

\subsection{Quadratic Lower Bounds for Layered Branching Programs}
\label[section]{sec:layered}

Here we briefly remark that if we assume that the branching program is \emph{layered}, we can obtain a truly quadratic $\Omega(n^2)$ lower bound (for a different function). Let 
\[
S_{n,d} (x_1, \ldots, x_{n}) = \sum_{\substack{B \subseteq [n] \\ |B|=d}} \prod_{i \in B} x_i
\]
denote the elementary symmetric polynomial of degree $d$. Chatterjee et al.\ \cite{CKSV22} proved that any \emph{layered} algebraic branching program computing $S_{n,n/10}$ over fields of characteristic $0$ has size $\Omega(n^2)$. This lower bound does not assume multilinearity (and their model even allows the edge labels to be arbitrary affine functions in $x_1, \ldots, x_n$). A key ingredient in the proof is an upper bound on the dimension of the variety cut by the first order partial derivatives of $S_{n,d}$, proved by \cite{MZ17, LMP19}. This upper bound was recently proved for any characteristic by Orzel \cite{Orzel25}, which implies a lower bound for ABPs over $\F_2$ (which is necessary for us, as we consider algebraic computations over $\F_2$).

Consider then a read-once \emph{layered} $\oplus$-BP of size $S$ computing $\bool{(S_{n,n/10})}$. Applying the transformation in \cref{cl:alg-to-bool}, we obtain a layered algebraic branching program computing $S_{n,n/10}$, which implies, by \cite{CKSV22, Orzel25}, that $S=\Omega(n^2)$. Note that \cref{cl:alg-to-bool} uses the fact that the $\oplus$-BP is read-once. We cannot omit this condition, even though the lower bound proof of \cite{CKSV22, Orzel25} does not require the ABP to be multilinear.

Finally, note that $h:=\bool{(S_{n,d})}$ is obviously an explicit function for any $d$, since for every $x \in \set{0,1}^n$, $h(x) = \binom{|x|}{d} \bmod 2$, where $|x|$ denotes the Hamming weight of $x$.

Chatterjee et al.\ \cite{CKSV22} also proved lower bounds for \emph{unlayered} branching programs, but these lower bounds are much weaker and not helpful for us in this context.

\section{A Polynomial-Time Dynamic Programming Algorithm}
\label[section]{sec:dp}

We now present a polynomial-time algorithm for computing $\bool{f}$, where $f$ is as defined in \cref{def:hard-poly}. To that end, we adopt a more combinatorial view of what $\bool{f}$ actually computes.

Consider again the elements of $\Z_n$ on a cycle and the process of constructing $\sigma_B$ in \cref{def:hard-poly}. We depict the action of matching $j$ to $\sigma_B(j)$ as drawing a directed chord between $j$ and $\sigma_B(j)$, labeled by $y_j$. We shall soon prove that these chords are always \emph{non-crossing}. An example is depicted in \cref{fig:cycle_graph}

% Define global styles outside the environment to protect the compiler parser
\tikzset{
    blackbullet/.style={circle, fill=black, minimum size=5pt, inner sep=0pt},
    bluebullet/.style={circle, fill=blue, minimum size=5pt, inner sep=0pt},
    dotsnode/.style={fill=white, inner sep=2pt, font=\footnotesize},
    every label/.style={font=\footnotesize}
}

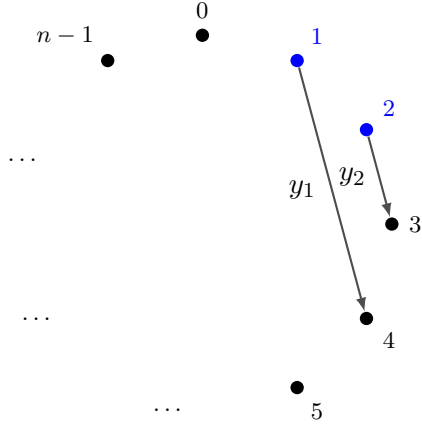
\begin{figure}[h]
\begin{center}

\begin{tikzpicture}

    \def\R{2.5} % Radius of the cycle graph

    % =========================================================================
    % VERTICES (Clockwise layout starting from 12 o'clock / 90 degrees)
    % =========================================================================
    \node[blackbullet, label={90:$0$}]                        (N0)   at (90:\R)   {};
    \node[bluebullet,  label={60:\textcolor{blue}{$1$}}]       (N1)   at (60:\R)   {};
    \node[bluebullet,  label={30:\textcolor{blue}{$2$}}]       (N2)   at (30:\R)   {};
    \node[blackbullet, label={0:$3$}]                         (N3)   at (0:\R)    {};
    \node[blackbullet, label={-30:$4$}]                        (N4)   at (-30:\R)  {};
    \node[blackbullet, label={-60:$5$}]                        (N5)   at (-60:\R)  {};
    
    % Ellipses filling the remaining perimeter (6, 7, 8, 9, 10 o'clock positions)
    \node[dotsnode] (Nd1) at (-100:\R) {$\dots$};
    \node[dotsnode] (Nd2) at (-150:\R) {$\dots$};
    \node[dotsnode] (Nd3) at (160:\R)  {$\dots$};
    
    % Last node before returning to 0 (11 o'clock position)
    \node[blackbullet, label={120:$n-1$}]                     (Nnm1) at (120:\R)  {};

    % =========================================================================
    % CHORDS / MATCHING ARCS
    % =========================================================================
    % Chord between 1 and 4 labeled y_1
    \draw[-latex, thick, black!70] (N1) -- (N4) 
        node[midway, left, xshift=-2pt, text=black] {$y_1$};
        
    % Chord between 2 and 3 labeled y_2
    \draw[-latex, thick, black!70] (N2) -- (N3) 
        node[midway, left, xshift=-1pt, text=black] {$y_2$};

\end{tikzpicture}
\caption{Cycle with chords. Blue nodes are in $B$ and black nodes are not in $B$. The matching $\sigma_B$ maps $2$ to $3$ by an edge labeled $y_2$ and then $1$ to $4$ by an edge labeled $y_1$.}
    \label[figure]{fig:cycle_graph}
\end{center}
\end{figure}

Given $x=(x_0, \ldots,x_{n-1}) \in \set{0,1}^n$ we think of $x$ as assigning bits on the $n$ elements of $\Z_n$. A set $B$, along with its non-crossing matching $\sigma_B$, is called \emph{$x$-valid} if for every $j \in B$, $x_j \neq x_{\sigma_B(j)}$. Similarly, given $y \in \set{0,1}^n$, we say that $B$ is \emph{$y$-eligible} if for all $j\in B$, $y_j$=1.

Note that since the additions and multiplications in \cref{def:hard-poly} are modulo $2$, for every input $x=(x_0, \ldots, x_{n-1})$, we have that $f_B(x)=1$ if and only if $B$ is $x$-valid. It follows that given an input $(x,y) \in \set{0,1}^n \times \set{0,1}^n$, the function $\bool{f}$ counts, modulo 2, the number of $x$-valid and $y$-eligible sets $B$. We will show that one can in fact count this number exactly in polynomial time (and therefore trivially compute its parity).

Suppose that instead of a \emph{cycle} of length $n$ we were to construct a non-crossing matching on an \emph{interval} of length $n$ in a similar fashion (one may think of this process as matching opening parentheses `(' with closing parentheses `)' in a well-matched parentheses sequence, but of course, on an interval not every subset $B$ of $n/2$ opening parentheses corresponds to a well-matched sequence). Ignoring the $y$ part for the time being, this setting naturally lends itself to a dynamic programming algorithm: we construct a table $M_{i,j}$ in which the $(i,j)$-th cell counts the number of $x$-valid matchings in the subinterval $[i,j]$, with our eventual goal being to compute $M_{0,n-1}$. To compute $M_{i,j}$, we use the fact that $i$ must be matched to some element in $[i+1, j]$, which gives the recursive formula
\[
M_{i,j} = \sum_{k \in [i+1, j], x_k \neq x_i} M_{i+1,k-1} \cdot M_{k+1, j}.
\]
(The ``$x_k \neq x_i$'' condition makes sure we only count $x$-valid matchings). The base cases for this induction are empty intervals whose value is $1$.

Our dynamic programming algorithm is inspired by this observation, but the fact that we are working with a cycle and not an interval means that some of the matchings can ``wrap around'' and are not accounted for by the formula above.

To solve this issue, we instead perform the count slightly differently. Consider again the set $B$ along with a matching $\sigma_B$. Let $\ell$ be the smallest element $j \in \Z_n$ such that $\sigma_B(j) < j$ when considered as integers (if there's no such element, set $\ell=0$). We say that $\ell$ is the \emph{leader} of the matching.
 
We now state and prove a useful combinatorial lemma.

\begin{lemma}
\label[lemma]{lem:matching-cut-cycle}
Let $B \subseteq \Z_n$ be a subset of size $n/2$ and consider the matching $\sigma_B$ as constructed in \cref{def:hard-poly}. Let $\ell$ be the leader of the matching (as defined above). ``Cut'' the cycle at $\ell$ so that we get the interval 
\[
\ell, \ell+1, \ell+2, \ldots, n-1, 0,1,2,\ldots, \ell-1
\]
For every $j \in B$, draw a directed edge between $j$ and $\sigma_B(j)$ on this interval (an illustration of this operation appears in \cref{fig:cut-cycle}). Then these edges are non-crossing, not wrapping around, and they all go from left to right.
\end{lemma}

\begin{figure}[h]
\begin{center}

\tikzset{
    blackbullet/.style={circle, fill=black, minimum size=5pt, inner sep=0pt},
    bluebullet/.style={circle, fill=blue, minimum size=5pt, inner sep=0pt},
    dotsnode/.style={fill=white, inner sep=2pt, font=\footnotesize},
    every label/.style={font=\footnotesize}
}

\begin{tikzpicture}

    % =========================================================================
    % PART 1: THE CYCLE GRAPH LAYOUT (TOP)
    % =========================================================================
    \def\R{2.5} % Cycle radius

    % Clockwise layout starting at 12 o'clock (90 degrees)
    \node[blackbullet, label={90:$0$}]                        (N0)   at (90:\R)   {};
    \node[bluebullet,  label={60:\textcolor{blue}{$1$}}]       (N1)   at (60:\R)   {};
    \node[bluebullet,  label={30:\textcolor{blue}{$2$}}]       (N2)   at (30:\R)   {};
    \node[blackbullet, label={0:$3$}]                         (N3)   at (0:\R)    {};
    \node[blackbullet, label={-30:$4$}]                        (N4)   at (-30:\R)  {};
    \node[blackbullet, label={-60:$5$}]                        (N5)   at (-60:\R)  {};
    
    % Ellipses along the bottom right
    \node[dotsnode] (Nd1) at (-90:\R) {$\dots$};
    
    % Node \ell at 7 o'clock (-120 degrees)
    \node[bluebullet,  label={-120:\textcolor{blue}{$\ell$}}]  (Nell) at (-120:\R) {};
    
    % Ellipses filling the left perimeter
    \node[dotsnode] (Nd2) at (-150:\R) {$\dots$};
    \node[dotsnode] (Nd3) at (150:\R)  {$\dots$};
    
    % Last node before returning to 0 (11 o'clock / 120 degrees)
    \node[blackbullet, label={120:$n-1$}]                     (Nnm1) at (120:\R)  {};

    % Cycle Directed Arcs
    % \ell to 0 (labeled on the right so it stays inside the circle)
    \draw[-latex, thick, black!70] (Nell) -- (N0) 
        node[midway, right, xshift=2pt, text=black] {$y_\ell$};

    % 1 to 4 (labeled on the left to stay inside the circle)
    \draw[-latex, thick, black!70] (N1) -- (N4) 
        node[midway, left, xshift=-2pt, text=black] {$y_1$};
        
    % 2 to 3 
    \draw[-latex, thick, black!70] (N2) -- (N3) 
        node[midway, left, xshift=-1pt, text=black] {$y_2$};

    % =========================================================================
    % PART 2: THE LINEAR PATH LAYOUT (BOTTOM)
    % =========================================================================
    \def\LY{-6.5} % Baseline Y-coordinate shifted down

    % Draw the continuous path baseline 
    \draw[thick, gray!30] (-5.8, \LY) -- (5.8, \LY);

    % Linear sequential nodes (Bullets with labels below)
    \node[bluebullet,  label={below:\textcolor{blue}{$\ell$}}] (L_ell)   at (-5.0, \LY) {};
    \node[dotsnode]                                            (L_d1)    at (-4.0, \LY) {$\dots$};
    \node[blackbullet, label={below:$n-1$}]                    (L_nm1)   at (-3.0, \LY) {};
    \node[blackbullet, label={below:$0$}]                      (L_0)     at (-2.0, \LY) {};
    \node[bluebullet,  label={below:\textcolor{blue}{$1$}}]    (L_1)     at (-1.0, \LY) {};
    \node[bluebullet,  label={below:\textcolor{blue}{$2$}}]    (L_2)     at (0.0,  \LY) {};
    \node[blackbullet, label={below:$3$}]                      (L_3)     at (1.0,  \LY) {};
    \node[blackbullet, label={below:$4$}]                      (L_4)     at (2.0,  \LY) {};
    \node[blackbullet, label={below:$5$}]                      (L_5)     at (3.0,  \LY) {};
    \node[dotsnode]                                            (L_d2)    at (4.0,  \LY) {$\dots$};
    \node[blackbullet, label={below:$\ell-1$}]                 (L_ellm1) at (5.0,  \LY) {};

    % Linear Directed Arcs (Nested perfectly without crossing)
    % 1. The outer arc from \ell to 0
    \draw[-latex, thick, black!70] (L_ell) to[out=45, in=135] 
        node[midway, above, text=black] {$y_\ell$} (L_0);

    % 2. Middle arc from 1 to 4
    \draw[-latex, thick, black!70] (L_1) to[out=55, in=125] 
        node[midway, above, text=black] {$y_1$} (L_4);

    % 3. Tight inner arc from 2 to 3
    \draw[-latex, thick, black!70] (L_2) to[out=70, in=110] 
        node[midway, above, text=black] {$y_2$} (L_3);

\end{tikzpicture}
\caption{``Cutting'' a cycle at the leader $\ell$ and drawing the edges on an interval.}
\label[figure]{fig:cut-cycle}
\end{center}
\end{figure}

\begin{proof}
If there's no element $j \in \Z_n$ such that $\sigma_B(j)<j$ then $\ell=0$. In this case, the interval equals $0,\ldots,n-1$, and the fact that the edges all go from left to right is rather obvious, as $\sigma_B(j) > j$ for all $j$. 
The proof of the non-crossing property is by induction on the construction of $\sigma_B$. We claim that at each stage, when we match $j$ to $\sigma_B(j)$, all elements in $[j+1, \sigma_B(j)-1]$ have already been matched. This is definitely true in the first stage, as we match $j$ to $j+1$ so that interval is empty. At any later step, we similarly match $j$ to its neighbor on the cycle $k$. If $k$ is a neighbor of $j$ it means that all elements in $[j+1, k-1]$ were already deleted and hence matched before.

Suppose now $\ell >0$ is a leader such that $\sigma_B(\ell) < \ell$, and order the elements as
\[
\ell, \ell+1, \ell+2, \ldots, n-1, 0,1,2,\ldots, \sigma_B(\ell), \ldots, \ell-1
\]
By definition there is a directed edge from $\ell$ to $\sigma_B(\ell)$. As before, since $\ell$ was connected to $\sigma_B(\ell)$ it means that all elements on the arc $[\ell+1, \sigma_B(\ell)-1]$ of the cycle were already matched to one another and erased from the cycle. Hence the elements in each of the intervals $[\ell+1, \sigma_B(\ell)-1]$ and $[\sigma_B(\ell)+1,\ell-1]$ are matched among themselves and the edge from $\ell$ to $\sigma_B(\ell)$ doesn't intersect any other edge, and no edge wraps around. With an identical argument we can argue that in each interval there are no intersecting edges.

To prove that each edge goes from left to right, consider a matched pair $j \in B$ and $\sigma_B(j) \notin B$. If $j=\ell$ this was already established. If $j \in [0,\ell-1]$, then by the definition of a leader we must have $\sigma_B(j)>j$ and clearly in the interval as ordered above the edge goes from left to right. The final case to consider is thus $j \in [\ell+1, n-1]$: suppose that $\sigma_B(j)$ resides to the left of $j$ in the ordering above, both lying in the interval between $\ell$ and $\sigma_B(\ell)$. In the construction of $\sigma_B$, we try to match an element with its neighboring element on the cycle, ordered clockwise. We could have only matched $\ell$ with $\sigma_B(\ell)$ if all the elements 
\begin{equation}
\label{eq:subinterval}
\ell+1, \ell+2, \ldots, n-1, 0, 1, \ldots, \sigma_B(\ell)-1
\end{equation}
were already matched among themselves, which means that when $j$ was matched to $\sigma_B(j)$, $\ell$ and $\sigma_B(\ell)$ were not erased yet. Since we match adjacent elements directed clockwise, we must have $\sigma_B(j)$ appearing to the right of $j$ in the subinterval \eqref{eq:subinterval}, as otherwise $j$ could not be adjacent to $\sigma_B(j)$ (since $\ell$ was not yet erased).
\end{proof}

Conversely, given any such cyclic shift of an interval with leader $\ell$ and labeled matched edges, we can uniquely recover the set $B$ and the mapping $\sigma_B$ using the direction of the edges (recall that an edge directed from $j$ to $k$ implies that $j \in B$ and $k\notin B$).

Consider now an \emph{interval} (rather than an arc) $[i,j]$, $i\le j$. We say that a set $B$ with a non-crossing matching $\sigma_B$ is \emph{external} to $[i,j]$ if $\ell$ is not in the interval $[i,j]$. Let $\Ext_{i,j}$ denote the number of possible ways to match the elements of $[i,j]$, using $x$-valid and $y$-eligible \emph{external} sets $B$.

Similarly, $B$ is \emph{internal} to $[i,j]$ if $\ell$ is inside the interval $[i,j]$, and let $\Int_{i,j}$ denote the number of possible ways to match the elements of $[i,j]$, using $x$-valid and $y$-eligible \emph{internal} sets $B$.

Since every $B$ is internal to $[0,n-1]$, we are interested in computing $\Int_{0,n-1}$.

We are now ready to prove the main theorem of this section.

\begin{theorem}
\label[theorem]{thm:counting-sets}
There exists a polynomial-time algorithm that, given $x,y \in \set{0,1}^n \times \set{0,1}^n$, computes the number of $x$-valid and $y$-eligible sets $B$.
\end{theorem}

\begin{proof}[Proof of \cref{thm:counting-sets}]
We compute $\Ext_{i,j}$ and $\Int_{i,j}$ by induction on the length of the interval. The length needs to be even for such matchings to exist, that is, if $i\le j$ and the length $(j-i+1)$ is an odd number then $\Ext_{i,j}=\Int_{i,j}=0$. Further, if the interval is empty, that is $j<i$, $\Ext_{i,j}=\Int_{i,j}=1$. 

For larger lengths, our first claim is the following:

\begin{claim}
\label[claim]{cl:external-matchings}
\begin{equation}
\label{eq:external}
\Ext_{i,j} = \sum_{k \in [i+1, j], x_i \neq x_k} y_i \cdot \Ext_{i+1,k-1} \cdot \Ext_{k+1, j}.
\end{equation}
\end{claim}

\begin{proof}[Proof of \cref{cl:external-matchings}]
Indeed, the equation goes over all possible ways to match the element $i$ with an element $k \in [i,j]$. Since we're only counting $x$-valid matchings, we only need to consider indices $k$ such that $x_i \neq x_k$. Since the leader $\ell$ is not in $[i,j]$, it is not in $[i+1,k-1]$ nor in $[k+1,j]$, so we multiply the relevant number of external matchings for these subintervals. Further, we claim that since the leader $\ell$ is not in $[i,j]$ it must be that $i \in B$ and $k \notin B$: if $\ell<i$, then by cutting the cycle at $\ell$ we obtain the interval
\[
\ell, \ell+1, \ldots, i, \ldots, k, \ldots, j \ldots, n-1, 0, 1, \ldots, \ell-1.
\]
By \cref{lem:matching-cut-cycle}, since the directed edges go from left to right, we see that if $i$ is matched to $k$ we must have $i \in B, k \notin B$. 

On the other hand, if $\ell > j$, then since $\ell$ is the smallest element with $\sigma_B(\ell) < \ell$, and $i < k < j < \ell$, we must have $k=\sigma_B(i)$ and thus $i \in B$ and $k\notin B$.

Since we established that $i \in B$ and $k \notin B$ in both cases, we multiply by $y_i$ to only count $y$-eligible matchings.
\end{proof}

Now consider internal matchings.

\begin{claim}
\label[claim]{cl:internal-matchings}
\begin{equation}
\label{eq:internal}
\Int_{i,j} = \sum_{k \in [i+1, j], x_i \neq x_k} \left( y_k \cdot \Int_{i+1,k-1} \cdot \Ext_{k+1, j} + y_i \cdot \Ext_{i+1,k-1} \cdot \Int_{k+1,j} \right).
\end{equation}
\end{claim}

\begin{proof}[Proof of \cref{cl:internal-matchings}]
Here we count matchings in which the leader $\ell$ is inside $[i,j]$. We again go over all possible elements $k$ that can be matched to $i$. We have either $\ell \in [i,k]$ or $\ell \in [k+1, j]$.

In the latter case ($\ell \in [k+1,j]$), by definition of $\ell$, and since $k < \ell$, the directed edge must go from $i$ to $k$ and therefore $i \in B, k\notin B$ (as otherwise $k$ would be a smaller element than $\ell$ with $\sigma_B(k) < k$, which contradicts the definition of $\ell$).  We thus multiply the number of external matchings on $[i+1, k-1]$ (since the leader is not in that interval) by the number of internal matchings on $[k+1, j]$, going over all $k$ such that $x_k \neq x_i$ (to only count $x$-valid matchings) and multiplying by $y_i$ (to only count $y$-eligible matchings). This accounts for the second term in \eqref{eq:internal}.

We are left with the case $\ell \in [i,k]$. In this case, when we cut the cycle at position $\ell$, 
\[
\ell, \ell+1, \ldots, k, \ldots, j, \ldots, n-1, 0, 1, \ldots, i, \ldots, \sigma_B(\ell), \ldots, \ell-1
\]
By \cref{lem:matching-cut-cycle} we see that the edge must be directed from $k$ to $i$, that is, $k \in B$ and $i=\sigma_B(k) \notin B$. We thus multiply the number of internal matchings on $[i+1, k-1]$ by the number of external matchings on $[k+1, j]$, going over all $k$ such that $x_k \neq x_i$ (to only count $x$-valid matchings) and multiplying by $y_k$ (to only count $y$-eligible matchings). This accounts for the first term in \eqref{eq:internal}.
\end{proof}

\cref{cl:external-matchings} and \cref{cl:internal-matchings} now establish \cref{thm:counting-sets}. As noted above, in order to compute $\Int_{0,n-1}$ we recursively compute $\Int_{i,j}$ and $\Ext_{i,j}$ for all $i<j$, by induction on the length of the interval. There are $O(n^2)$ quantities to compute, and using \eqref{eq:external} and \eqref{eq:internal}, each can be computed in time $O(n)$.
\end{proof}

\section{Algebraic Circuits Lower Bound for a Polynomial in $\VP$}
\label[section]{sec:VP}

Our proof from \cref{sec:dp} also implies that the polynomial $f$ from \cref{def:hard-poly} is in $\VP$, the class of polynomial families of degree $\poly(n)$ and circuits of size $\poly(n)$ (in \cite{AKV20}, it is only claimed to be in $\VNP$). This shows that the lower bound of \cite{AKV20} also holds for a polynomial in $\VP$. We remark that the technical condition that Alon et al.\ \cite{AKV20} need $f$ to satisfy is that its coefficient matrix is full rank under any partition of the variables. This technique was introduced by Raz \cite{Raz09} and was later also used in \cite{Raz06, RY08}, to name only a few examples. A more systematic study of this technique appears in \cite{FLSY26}. We do not go into details here and refer to any of these papers for precise definitions. We call such a polynomial a \emph{full rank polynomial}.

A full rank polynomial in $\VP$ was already constructed in \cite{RY08}. For technical reasons, however, in the proof one needs to consider such a polynomial $f(x,y) \in \F_2 [x,y]$ as a polynomial in $x$ over the field $\F(y)$ (here $x,y$ are vectors of variables), and then the rank is computed over $\F(y)$. In the construction of Raz and Yehudayoff \cite{RY08}, the number of variables in $y$ is $O(n^3)$, whereas the lower bound of \cite{AKV20} is nearly-quadratic in the number of variables in $x$. Therefore, if the complexity is measured as a function of the total number of variables, the lower bound is meaningless for the polynomial of \cite{RY08} (in $f$ from \cref{def:hard-poly} the number of variables in $y$ is $n$, so this problem doesn't arise). When $\F$ is large enough, one can take the construction of Raz and Yehudayoff \cite{RY08} and plug in random values to the $y$ variables. With high probability, after this fixing, one obtains a full-rank $n$-variate polynomial over $\F$. This construction, however, is not explicit (and requires large fields). For further discussion on this topic see Section 4 of \cite{AKV20}.

Fortunately, this somewhat annoying issue is no longer an issue, since we can prove:

\begin{theorem}
\label[theorem]{thm:hard-polynomial-VP}
Let $f$ be as in \cref{def:hard-poly}. Then $f$ has a circuit of size $O(n^3)$ (and in particular, $f \in \VP$).
\end{theorem}

\begin{proof}
We construct a circuit following the proof of \cref{thm:counting-sets}, using equations similar to \eqref{eq:external} and \eqref{eq:internal}. For every $i\ \le j$ such that $j-i+1$ is even we add two gates $E_{i,j}$, $I_{i,j}$ and connect them as follows: 
\begin{align*}
E_{i,j} &= \sum_{k \in [i+1,j]} y_i \cdot (x_i + x_k) \cdot E_{i+1,k-1} \cdot E_{k+1, j} \\
I_{i,j} &= \sum_{k \in [i+1,j]}  \Big( y_k \cdot (x_i + x_k) \cdot I_{i+1, k-1} \cdot E_{k+1,j} + y_i \cdot (x_i + x_k) \cdot E_{i+1, k-1} \cdot I_{k+1,j}  \Big)
\end{align*}
(where each gate $E_{i',i''}$ is understood to be the constant $1$ if $i''<i'$, and similarly for $I_{i',i''}$). The output of the circuit is the gate $I_{0,n-1}$.

The circuit $C$ is multilinear: it follows by induction on $j-i+1$ that $E_{i,j}$ and $I_{i,j}$ are multilinear since the subcircuits rooted at them are only connected to $x$ and $y$ variables with indices in $[i,j]$.

One can prove by induction that $C$ computes $f$. A different way to see it is by directly reducing to \cref{thm:counting-sets}. Since the circuit $C$ agrees with $\bool{f}$ functionally on $\set{0,1}^n \times \set{0,1}^n$, and since it is multilinear, it must compute the polynomial $f$.

The circuit $C$ has size $O(n^3)$: the fan-in of each of the $O(n^2)$ gates $I_{i,j}$ and $E_{i,j}$ is $O(n)$, for a total of $O(n^3)$ edges. One can then convert this circuit to a bounded fan-in circuit with $O(n^3)$ gates.
\end{proof}

Note that the polynomial in \cref{def:hard-poly} is defined over any field and \cref{thm:hard-polynomial-VP} is true over any field.

\cref{thm:hard-polynomial-VP} shows a barrier for the technique of analyzing the rank of the coefficient matrix under various partitions: it can't prove lower bounds beyond $\Omega(n^3)$ (strictly speaking, this also follows from the randomized construction mentioned above using the polynomial of Raz and Yehudayoff \cite{RY08}. The \emph{existence} of a circuit of size $O(n^3)$ computing a full-rank polynomial, even non-explicitly, is enough to prove the barrier result).

\section{Open Problems}

One could hope, of course, to solve \cite[Research Problem 16.14]{Juk12} completely and prove super-polynomial lower bounds for read-once parity branching programs. Such a result would follow from super-polynomial lower bounds for syntactically multilinear algebraic branching programs. A study of the limitations of current techniques for proving such lower bounds was recently initiated by Fabris et al. \cite{FLSY26}. As a first step, we propose proving a cubic $\Omega(n^3)$ lower bound, perhaps by proving such a lower bound on syntactically multilinear algebraic circuits: that would prove that the construction in \cref{thm:hard-polynomial-VP} is optimal, but it's worth mentioning that we have no strong reasons to believe that it is indeed optimal, or that there isn't another full-rank polynomial with a circuit of size $O(n^2)$ (by the results of \cite{AKV20}, such a construction \emph{would} be optimal, up to logarithmic factors).

More generally, for many boolean models of computation, the best lower bounds known are proved using Nechiporuk's \cite{Nec66} method. It is interesting to try and find more cases in which stronger lower bounds can be proved using various methods, in particular using  reductions to lower bounds in algebraic circuit complexity.

\bibliographystyle{alphaurlpp}
\bibliography{references}

\end{document}